\documentclass{article}
\usepackage{graphicx} 
\usepackage{amsthm,amssymb,amsfonts,amsmath}
\usepackage{fullpage}
\usepackage{enumitem}
\usepackage{hyperref}
\usepackage{comment}
\usepackage{xspace}
\usepackage{url}
\usepackage{hyperref}
\usepackage{cleveref}
\usepackage[usenames,dvipsnames]{xcolor}

\newtheorem{theorem}{Theorem}
\newtheorem{lemma}[theorem]{Lemma}
\newtheorem{observation}[theorem]{Observation}
\theoremstyle{definition}
\newtheorem{definition}[theorem]{Definition}

\newcommand{\nestingbirdbox}{\textsc{Nesting Bird Box}\xspace}
\newcommand{\artgallery}{\textsc{Art Gallery}\xspace}

\newcommand{\boxx}{\text{birdhouse}\xspace}
\newcommand{\boxxes}{\text{birdhouses}\xspace}
\newcommand{\Vis}{\ensuremath{\mathrm{Vis}}}

\newcommand{\Z}{\ensuremath{\mathbb{Z}}\xspace}
\newcommand{\R}{\ensuremath{\mathbb{R}}\xspace}
\newcommand{\ER}{\ensuremath{\exists\mathbb{R}}\xspace}
\newcommand{\NP}{\ensuremath{\text{NP}}\xspace}
\newcommand{\PSPACE}{\ensuremath{\text{PSPACE}}\xspace}
\newcommand{\Polytime}{\ensuremath{\text{P}}\xspace}

\title{The Nesting Bird Box Problem is \ER-complete:
\\ Sharp Hardness Results for the Hidden Set Problem}

\author{
  Lucas Meijer\thanks{L.~Meijer was generously supported by the Netherlands Organisation for Scientific Research (NWO) under project no. VI.Vidi.213.150.} \\
  Utrecht University \\
  \texttt{l.meijer2@uu.nl}
  \and
  Tillmann Miltzow\thanks{T.~Miltzow was generously supported by the Netherlands Organisation for Scientific Research (NWO) under project no. VI.Vidi.213.150.} \\
  Utrecht University \\
  \texttt{t.miltzow@uu.nl}
  \and
  Johanna Ockenfels \\
  ETH Zürich\\
  \texttt{jockenfels@ethz.ch}
  \and
  Milo\v{s} Stojakovi\'{c}\thanks{M.~Stojakovi\'c was partly supported by the Science Fund of the Republic of Serbia, Grant \#7462: Graphs in Space and Time: Graph Embeddings for Machine Learning in Complex Dynamical Systems (TIGRA), and partly supported by the Ministry of Science, Technological Development and Innovation of the Republic of Serbia (grants 451-03-33/2026-03/200125 \& 451-03-34/2026-03/200125).} \\
  Department of Mathematics and Informatics, \\ Faculty of Sciences, University of Novi Sad, Serbia \\
  \texttt{milosst@dmi.uns.ac.rs}
}

\begin{document}

\maketitle
\begin{abstract}
    In the (Nesting) Bird Box Problem we are given a polygonal domain $P$ and a number $k$ and we want to know if there is a set $B$ of $k$ points inside $P$ such that no two points in  $B$ can see each other.
    The underlying idea is that each point represents a \boxx and many birds only use a \boxx if there is no other occupied \boxx in its vicinity.     We say two points $a,b$ see each other if the open segment $ab$ intersects neither the exterior of $P$ nor any vertex of $P$.

    We show that the \nestingbirdbox problem is \ER-complete.
    The complexity class \ER can be defined by the set of problems that are polynomial time equivalent to finding a solution to the equation
    $p(x) = 0$, with $x\in \R^n$ and $p\in \Z[X_1,\ldots,X_n]$.
    The proof builds on the techniques developed in the original  \ER-completeness proof of the \artgallery problem.
    However our proof is significantly shorter for two reasons.
    First, we can use recently developed tools that were not available at the time.
    Second, we consider polygonal domains with holes instead of simple polygons.
        
\end{abstract}

\vfill

\begin{center}
    \includegraphics[width=0.3\textwidth]{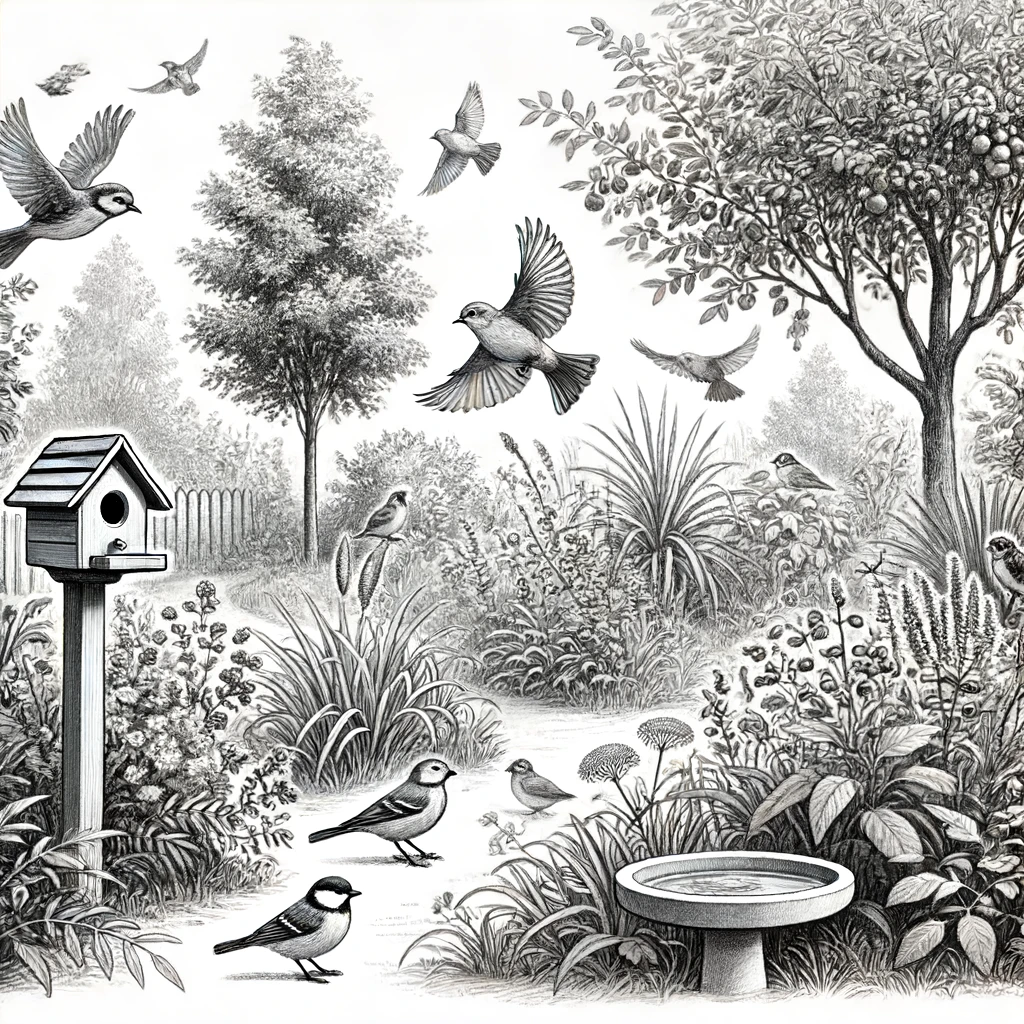}
    
    Many birds enjoy life in a beautiful garden.
    (Generated by ChatGPT.)
\end{center}

\vfill

\newpage
\section{Introduction}
\textbf{Tangible Example.}
    It is March, and a magnificent magnolia tree is showcasing its white flowers in a beautiful garden in the countryside of Denmark. Only birds could make this paradise even better. After a quick discussion, the decision is made: bird feeders are placed, a bird bath is installed, and nesting \boxxes are built.
However, the question remains: where should the nesting \boxxes be located?
The garden is expansive, with many potential sites for the \boxxes. An important consideration is that many birds will not nest in a \boxx if they can see another \boxx. 
(Although cute, birds are very territorial.)
Therefore, to maximize the number of birds taking shelter in the garden, we must ensure that no two \boxxes are within each other's line of sight.

In this paper, we investigate how to install the largest number of \boxxes such that no two can see each other, approached from a computer science perspective.
We show that the \nestingbirdbox problem is \ER-complete.
A problem is \ER-complete if it is polynomial time equivalent to solving the equation $p(x) = 0$, with $x\in \R^n$ and $p$ being a polynomial with integer coefficients.

\textbf{Motivation.}
One of the most popular problems in computational geometry is the \artgallery problem. 
The \artgallery problem is popular due to its appealing metaphor and many variants, and the beautiful proof by Steve Fisk~\cite{Fisk78a}.
In the \artgallery problem, we are working on the dominating set problem in the visibility graph of polygons.
The \nestingbirdbox problem is the independent set problem on the same graph class.
From a mathematical perspective both dominating set and independent set are (maybe equally) natural and important.
Another contribution of this paper is to popularize the \boxx metaphor for the independent set problem of the visibility graph, which to the best of our knowledge comes from Mikkel Abrahamsen.
In this work, we settle the algorithmic complexity of the \nestingbirdbox problem on polygonal domains.

\subsection{Definition and Results.}
A polygonal domain $P$ is a closed subset of the plane $\R^2$, without curved boundaries.
\begin{definition}[\nestingbirdbox Problem]
    Given a polygonal domain $P$ and an integer $k$, are there $k$ points within $P$ such that no two of the $k$ points are visible to each other?
\end{definition}
Typically, visibility is defined as follows: two points $a,b\in P$ are visible if and only if 
the segment $ab$ is fully contained in the polygonal domain $P$.
We refer to this version of visibility as \textit{ordinary visibility}.
In our case, we say two points $a,b\in P$ see each other if the open segment
$ab$ is contained in the polygonal domain and contains no vertex of $P$.
We refer to this version of visibility as \textit{open-visibility}. The idea behind open-visibility is that visibility regions do not contain their boundary fully.
(The \textit{visibility region} of a point $a$, denoted by $\Vis(a)$, is the set that can be seen by $a$, i.e., $ \{b \in P : a \text{ is visible to } b\}$.)
Open-visibility is less common than ordinary visibility. 
At the same time, our proofs crucially rely on this variant of visibility.
This is comparable to how the \ER-hardness of the \artgallery problem crucially relies on boundary conditions of ordinary visibility~\cite{S23c,AAM22}.
This definition of visibility is comparable how \textsc{Graph In Polygon} relates to the \textsc{Partial Drawing Extensibility}~\cite{LMM22}.
The given polygon in the \textsc{Graph In Polygon} can be seen as a pre-drawn graph.
However, in \textsc{Graph In Polygon}, the authors allowed newly drawn edges of the graph to lie on top of edges of pre-drawn graph, and for vertices of the pre-drawn graph to lie on the newly drawn parts.

We are now ready to state our main result.
\begin{theorem}\label{thm:birds}
    The \nestingbirdbox problem is \ER-complete on polygonal domains.
\end{theorem}

The main contribution is to settle the algorithmic complexity of this problem. 
Our proofs turn out to be significantly simpler than the 
proof that the \artgallery problem is \ER-complete~\cite{S23c,AAM22}.



In order to understand the full impact and the shortcomings
of our results, we will discuss it from different perspectives in the
following section.

\subsection{Discussion}
\textbf{Visibility Definition.}
Possibly the biggest shortcoming of Theorem~\ref{thm:art} is that it relies on a, mathematically speaking, somewhat unusual definition of visibility. 
In our proofs, it is crucial to be able to enforce a \boxx to lie on a specific line segment.
If we were to use the ordinary definition of visibility, every optimal placement of \boxxes could be slightly perturbed and the solution would still be valid.
In other words, the solution space would be open.

We want to remark that the \ER-hardness proof of the \artgallery problem also strictly relies on the closedness of the solution space, which allows us to enforce guards to lie on a specific segment in a similar way.
So it is also an open research question what is the algorithmic complexity of the \artgallery problem with our definition of visibility.
Similar problems appear when we consider drawing a graph in a polygonal domain~\cite{lubiw2018complexity}.

Another remark is that changing the notion of visibility in the context of the \artgallery problem is a common phenomenon,
and in comparison our modification of visibility is arguably subtle.
    

\textbf{Similarity to \artgallery Proof.}
We want to point out that many parts of this proof are 
reminiscent of 
the \ER-hardness of the \artgallery problem.
This is not a complete coincidence; many gadgets can be easily 
adapted to work for the \nestingbirdbox problem. 

\textbf{Proof Simplicity.}
The biggest strength of this paper is the simplicity of our proofs.
Specifically, the proof is easy enough to be taught in a graduate
course in Computational Geometry, compared to the length of the original paper 
by Abrahamsen, Adamaszek and Miltzow.
There are three main reasons for this simplification.

\begin{itemize}
    \item We can reduce from continuous constraint satisfaction problems, which did not exist at the time.
    It gives much more flexibility and is easy to apply. Specifically, we do not need to find a gadget encoding $x\cdot y = 1$.
    \item We consider polygonal domains, i.e., polygons with holes. 
    This simplifies our proofs significantly as we can easily isolate gadgets from
    one another, by simply adding holes between different gadgets in the appropriate way.
    Copying variables to isolated locations was the main technical difficulty
    in the original paper.
    \item At last, we learned from previous experiences in the recent literature and streamlined the arguments.
\end{itemize}

\textbf{Relevance of \ER.}
It is natural to ask what is the added benefit of knowing that a problem is \ER-complete if we already
know that the problem is NP-hard.
There are several answers to this. 
The first one is that we are interested in this question purely from a mathematical perspective
and having a complete classification of the difficulty of this natural problem is desirable.
Another answer is that it helps calibrate our expectation management in terms of algorithmic
results that we can expect.
For example, if we want to solve the \nestingbirdbox problem using integer linear programming, we need to make concessions on their performance. 
There are some instances of the \nestingbirdbox problem that cannot be encoded as an integer linear program unless $\NP = \ER$. 
Thus we would expect that such encodings are only possible if some additional assumptions are met.
Vision-stability, as introduced by Hengeveld and Miltzow, could be one such assumption~\cite{HM21}.

\textbf{Simple Polygons.}
Another shortcoming of our result is that it applies only to polygons with holes.
We conjecture that the \nestingbirdbox problem is also \ER-complete for simple polygons.
However, we expect the proof would be significantly more complicated, as this proof uses holes in the polygon as a tool to ensure \boxxes will never be able to see each other unless required for a gadget.

\subsection{Related Work -- Nesting Bird Boxes}
The oldest reference we found that studies the ``visual independence'' traces back to 1970 by Kay and Guay~\cite{kay1970convexity}.
They were interested in sufficient conditions such that a set is the finite union of convex or star-shaped sets. 
Most importantly, if a polygon can be covered with $k$ convex shapes then we cannot place more than $k$ \boxxes in the polygon, as any two \boxxes in the same convex set can see each other.
(This is modulo our definition of visibility.)

In the literature the \nestingbirdbox{} problem is typically referred to as the \textsc{Hidden Set} problem in polygons.
The first algorithmic study was - to the best of our knowledge - in the PhD thesis by Shermer in 1989~\cite{shermer1989visibility}.
Shermer showed that the problem is NP-hard.
Furthermore, they studied various other combinatorial variations.
For example, Shermer showed that there are polygons that do not admit a hidden vertex guard set.
A hidden vertex guard set is a set of vertices that are mutually invisible to one another and see the entire polygon.

Eidenbenz studied the \nestingbirdbox problem in 1999 with ordinary visibility in polygons with and without holes and terrains~\cite{E02, eidenbenz1999many}.
Their motivation was that real estate agents want to sell a plot of land such that no two cabins built on this land can see each other.
Their first result concerns the vertex variant in polygonal domains with $n$ vertices, where all bird boxes have to be placed at vertices.
They showed that it cannot be approximated within a factor better than
$n^{1/6 - \varepsilon}/4$ for all $\varepsilon >0$, unless $\Polytime = \NP$.
The underlying idea is to reduce from the independent set problem in graphs.
Furthermore, they showed that the problem is APX-hard in simple polygons (without holes).

Eidenbenz also studied in 2006 the \textsc{Minimum Hidden Guard Set} problem
on simple polygon with $n$ vertices~\cite{eidenbenz2006finding}.
Here, they allow guards to be placed anywhere in the interior of the polygon.
They showed that this problem cannot be approximated better than $n^{1-\varepsilon}$ for all $\varepsilon > 0 $, unless $\Polytime= \NP$.
They also exhibit an algorithm that has an approximation factor of $n$, which makes their results tight.

Browne et al.~\cite{browne2023constant} gave a constant factor $(1/8)$-approximation algorithm for the \nestingbirdbox problem in 2023.
To do this, they first obtain a $(1/2)$-approximation for the \nestingbirdbox problem in weakly visible polygons and then use a result that any polygon can be decomposed into four types of weakly visible polygons to extend the result to general polygons.

For simple polygons, Alegr\'ia et al.~\cite{alegria20191} showed in 2019 that one can  compute a $(1/4)$-approximation for the \nestingbirdbox problem in $O(n^2)$ time.

For the special case where the polygonal region is restricted to be a thin grid $n$-gon, Bajuelos et al.~\cite{bajuelos2007solving} showed that no such polygon can fit more than $\lceil n/4\rceil$ \boxxes.
Here, a grid $n$-gon is a simple polygon with $n$ vertices with orthogonal, non-collinear edges that can be placed with its vertices on a $(n/2)$ by $(n/2)$ square grid.
A thin grid $n$-gon has the additional constraint that, when extending line segments from its reflex vertices, it cuts the polygon into the least possible amount of separate regions.

Finally, we note a link to twin-width; surprisingly, \nestingbirdbox on simple polygons is FPT when parametrized by twin-width~\cite{bonnet2022twin}.

For a more in-depth survey about the \nestingbirdbox problem, we refer to Browne's 2024 survey~\cite{browne2024overview}.





\subsection{The Existential Theory of the Reals}
The complexity class \ER can be defined in different ways. 
Perhaps the easiest way is as follows: \ER-complete problems are those that are polynomial time equivalent to finding a solution to 
$p(x) = 0$, with $x\in \R^n$ and $p$ a polynomial with integer coefficients.
It is known that \[\NP \subseteq \ER \subseteq \PSPACE.\]
While the first inclusion is trivial (it follows directly from a different definition of \ER) the second one is very intricate and goes back to Canny~\cite{C88, C88b}.

This complexity class emerged from two different traditions.
The first tradition comes from structural complexity theory. 
Researchers tried to formulate models of computation that capture real-valued computations. 
This culminated in the so-called BSS-model of computation~\cite{BCSS98, BS89}.
In a nutshell it generalizes Turing machine to work with real numbers.
However, within this tradition the complexity class \ER played a minor role.
The second tradition followed the results by Mn\"{e}v and Shor 
who originally cared for certain topological phenomena and precision issues
in a geometric context~\cite{S91, M88}.
In 2010, Marcus Schaefer suggested the name \ER~\cite{S10}, as he noticed that it encaptures many fundamental problems in computational geometry.

The main importance of the complexity class \ER 
stems from the fact that a large number of important algorithmic problems from many 
different areas are complete for it.
A compendium from Cardinal, Miltzow and Schaefer~\cite{ERcompendium}
gives an up to date overview.
Miltzow wrote lecture notes to acquaint the reader with \ER~\cite{M26} and Matou\v{s}ek wrote very nice lecture notes that focus on explaining
the \ER-hardness result of stretchability~\cite{M14}.

The main areas for \ER-complete problems are related to
intersection graphs~\cite{KM94, CFMTV18}, visibility~\cite{CH17}, graph drawing~\cite{FKMPTV23, LMM22, AKM23}, game theory~\cite{BM16}, polynomials~\cite{SS17}, polytopes~\cite{RG99, DHM19}, matrix and tensor ranks~\cite{S17},
machine learning, Markov decision processes~\cite{JJK22} and logic~\cite{vdZBL23}.
Problems that raised particular interest is \ER-completeness of the \artgallery problem~\cite{AAM18, S23c}, finding Nash Equilibrium~\cite{SS17},
the 4-dimensional Steinitz-problem~\cite{RG96}, training neural networks~\cite{BEMMSW22, AKM21, Z92}, and tensor rank~\cite{SS18}.
In our context, the \ER-completeness of convex covering is very related~\cite{A22}, due to the relationship between the \boxx number and the cover number.


\section{\ER-membership}
\label{sec:membership}
There are two ways to show \ER-membership. 
We can either describe an ETR-formula that encodes 
the \nestingbirdbox problem or we describe a polynomial time verification algorithm.
For educational purposes we do both.

\textbf{ETR-formula.}
We encode each \boxx using two variables.
Furthermore, we need a predicate $\texttt{Non-Visibility}$ that returns yes,
if and only if two points cannot see each other.
Then the formula has the form:

\[ \exists x_1,y_1,\ldots, x_k,y_k :  \bigwedge_{i,j} \ \texttt{Non-Visibility}(x_i,y_i , x_j,y_j)\]

It remains to explain how we can check for two points $p = (x_i,y_i)$ and $q = (x_j,y_j)$ that they cannot see each other.
To this note that it is sufficient to check if the line segment $pq$ contains a vertex of $P$ or crosses an edge of $P$.
Both of those tests can be formulated with the help of the orientation and collinearity test.
In the collinearity test, we test if three points $a,b,c$ are collinear, by evaluating the determinant
\[\det \begin{pmatrix}
a_x & b_x & c_x \\
a_y & b_y & c_y \\
1 & 1 & 1\end{pmatrix} = 0.\]

The orientation test asks if the three points $(a,b,c)$ in this order are in clockwise or counter clockwise orientation.
In other words, is the point $c$ to the left or to the right of the line $\ell(a,b)$.
This can be tested using the sign of the determinant

\[\det \begin{pmatrix}
a_x & b_x & c_x \\
a_y & b_y & c_y \\
1 & 1 & 1\end{pmatrix}.\]

Point $p$ can see point $q$ if and only if no vertex of the polygon is on the line segment $pq$ between $p$ and $q$ and no edge of the polygon intersects $pq$.
We use the collinearity test to check whether a vertex of the polygon lies on $pq$.
Likewise, we use the orientation test to check whether an edge $uv$ with endpoints $u$ and $v$ of the polygon intersects $pq$.
We see that $pq$ intersects $uv$ if and only if $u$ and $v$ lie on different sides of $pq$ and $p$ and $q$ lie on different sides of $uv$.
As such, we can use $\texttt{and}$ and $\texttt{or}$ to use the two smaller primitives in order to describe the $\texttt{Non-Visibility}$ predicate.

In general, the membership proof using ETR-formulas is either fairly tedious, and/or it relies on the familiarity of the reader with ETR-formulas. 
Frankly speaking, this experience can be very mixed. 
Some readers might have not seen ETR-formulas before reading this paper and others might find this description trivial as they have done many similar arguments uncountable many times.

We give therefore another \ER-membership proof that replaces experience with ETR-formulas with experience with real RAM algorithms. 

\textbf{Polynomial Verification.}
Another way to prove \ER-membership is to use a machine model of \ER. 
We use the model by Erickson, Hoog and Miltzow~\cite{EvdHM20}.
They show that a problem is in \ER, if and only if there exists a polynomial time verification algorithm,
which is allowed to run on the realRAM and accepts as a real-valued witness.
In our case the placement of all the boxes is the witness and checking that no two \boxxes see each other is the verification algorithm.
We can check whether two \boxxes see each other by computing whether the line segment between the \boxxes intersects any edge of the polygon, akin to the Visibility predicate from Efrat and Har-Peled~\cite{EFRAT2006238}.
%

\section{Overview of the Construction}
\label{sec:overview}
This section is dedicated to the proof of the \ER-hardness part of Theorem~\ref{thm:birds}.

\subsection{Reduction from Continuous Constraint Satisfaction Problems.}
\label{subsec:ccsp}
To show \ER-hardness, we reduce from a specific version of a continuous constraint satisfaction problem.
To avoid overly heavy formalism about constraint satisfaction problems, we merely define the problem at hand that we need
and refer to the paper of Miltzow and Schmiermann for the formal details~\cite{MS24}.
We define ETR($h$) as the following problem.
We are given variables 
$x_1,\ldots,x_n$ and constraints $C_1,\ldots,C_m$ of the form 
\[x+y = z, \, h(x) = y, \, x \geq 0, \, x =1.\]
We need to decide if there are variable values from the interval $[-1,1]$ such that all constraints are satisfied.

Miltzow and Schmiermann showed that ETR($h$) is \ER-hard when $h$ is well-behaved and curved~\cite{MS24}.
(See Theorem 1.11, with $\delta = 1$ and $f = h(x) - y$, then see Lemma 1.16; note that being triple algebraic is only needed for \ER-membership.)
By Miltzow and Schmiermann, we say that $h: I \rightarrow \R$ is \textit{well-behaved} if it satisfies the
following conditions. Firstly,
$h$ is a three-times continuously differentiable function, where $I\subset \R$ being an interval with $0$ in its interior.
Secondly, $h$ is a function of the form $h(x) = \frac{ax+b}{cx+d}$ for rational $a, b, c, d$.
Finally, we say $h$ is curved if it is not affine.

\subsection{Reduction Overview}
\label{sub:overview}
In the reduction from ETR($h$), we construct a polygonal domain, as displayed in \Cref{fig:overview}.
The polygonal domain consists of a big central chamber padded with various gadgets.

\begin{figure}
    \centering
    \includegraphics[page=12]{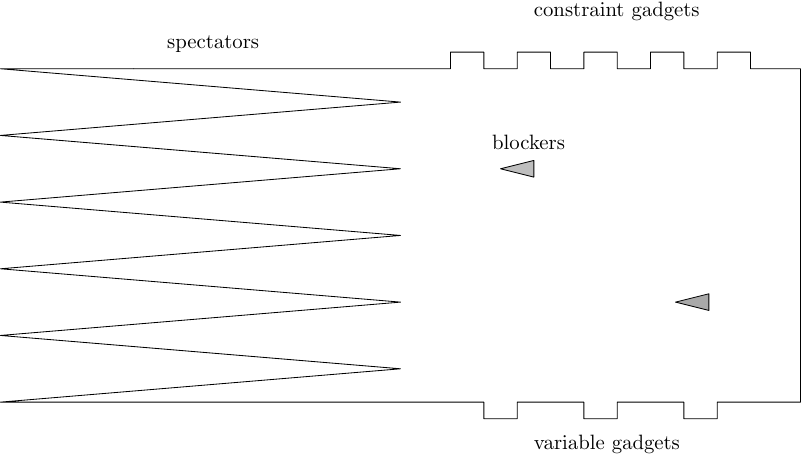}
    \caption{An overview of our construction. The variable gadgets at the bottom encode variables. The constraint gadgets at the top encode addition, scaling, and curvature. The spectators ensure no additional \boxxes can be added and the blockers ensure that constraint gadgets only interfere with its corresponding variables. }
    \label{fig:overview}
\end{figure}

\textbf{Variable Encoding.}
The first gadget ensures that there is a single \boxx{} at some undetermined position on a specific line segment (called \textit{\boxx segment}) in the interior of the polygonal domain. 
This effectively encodes a variable and the position of the \boxx on the \boxx segment corresponds to the value of the variable.
We call the \boxx on the \boxx segment \textit{variable \boxx}.
In our construction, the gadgets encoding variables will appear on the bottom of the central chamber.

\textbf{Constraint Encoding.}
The second idea is to have constraint gadgets that interact with the variable gadgets.
Two or three variable \boxxes can see into the gadget and if their variable encoding satisfies the intended constraint then
one extra \boxx can be placed and otherwise not.
The constraints that we will describe are
copy and scaling gadgets and
two addition gadgets ($x+y\leq z$ and $x+y \geq z$), 
two curved gadgets ($h(x) \leq y$ and $h(x) \geq y$).
A big part of the proof is to describe those gadgets and to show that they work correctly.
In the final construction, these gadgets will appear to the top of the central chamber.

\textbf{Blockers.}
The base of our construction is a big rectangle, see Figure~\ref{fig:overview}.
We attach all the variable gadgets to the bottom and all the constraint gadgets to the top.
Furthermore, we add small holes in a way that every variable \boxx can see exactly those constraint gadgets that it is supposed to interact with.
Those blockers are of triangular shape.
We describe in \Cref{sub:synthesis} the precise placement of the blockers.

\textbf{Spectators.}
We need to make sure that we only place \boxxes at the intended positions. 
For this purpose, we create little nooks where we know that a \boxx will be placed in some optimal solution.
Those \boxxes are called spectators and their sole purpose is to prevent that any \boxx is placed in their visibility region.
We use spectators both inside the gadgets as well as in the big rectangle.
Section~\ref{sub:spectators} formulates some general lemmas that will be helpful to use spectators.

\textbf{Threshold.}
At last, we will have to determine a threshold for the \nestingbirdbox problem.
We will do this by determining a threshold $t(G)$ for each gadget $G$ and the middle part individually and then take the sum of all of those.
At the same time, we will give a covering with $t(G)$ convex regions for each gadget $G$. 
As each convex region contains at most one \boxx{}, we also know an upper bound on the number of \boxxes per gadget. 
And thus we know that if there is any valid solution with the number of \boxxes equal to the sum of those thresholds, then each gadget $G$ must contain exactly $t(G)$ \boxxes.

\subsection{Spectators}
\label{sub:spectators}

\begin{figure}[t]
    \centering
    \includegraphics{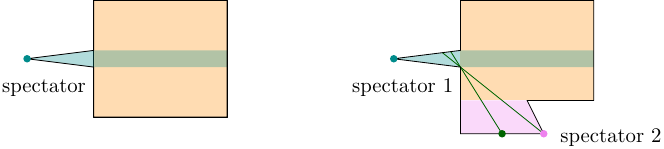}
    \caption{Left: The polygon is covered with two fully visible regions. 
    Each region can contain at most one \boxx. 
    Furthermore, the turquoise region has a vision minimal position in the small nook.
    We will argue that this will be the position of the spectator.
    Right: The second spectator is also vision minimal once we remove the visibility region of spectator~$1$.}
    \label{fig:spectator-idea}
\end{figure}

We call a \boxx a spectator if we can argue that there must be a fixed location for this \boxx in any solution.
In particular, any \boxx on a \boxx segment is not a spectator, such as the \boxxes that encode our variables.
We aim for the visibility polygons of the spectators to cover most of the polygon, all but the \boxx segments and specific regions of the constraint gadgets. 

We illustrate the general idea of spectators using Figure~\ref{fig:spectator-idea}.
On the left, there is a small nook of the polygon from where we can see very little.
Intuitively, it seems a good idea to place a \boxx there.
To make this mathematically more rigorous, we see that the polygon has a cover with two ``fully visible'' regions. 
Each region can contain at most one \boxx. 
Furthermore, the turquoise region has a vision minimal position in the small nook.
So, if there is any optimal solution to the \nestingbirdbox problem with two \boxxes within the polygon, we can find another one where we move the \boxx in the turquoise region to the spectator position.
Unfortunately, this idea is not enough, as can be seen on the right of Figure~\ref{fig:spectator-idea}.
As we can see, spectator~$2$ is not at a vision minimal position.
However, we already know the position of spectator~$1$ and once we ignore the area seen by spectator~$1$, the position of spectator~$2$ is vision minimal with respect to the rose region.
In the following, we will make those ideas more precise.

We assume that we are given a closed polygonal domain $P$.
We say that a region $C$ of the polygon $P$ is \textit{fully visible} if any two \boxxes placed in~$C$ can see each other by open-visibility. 
Note that not all convex regions have this property because visibility regions may be open. 
Fully visible regions are useful to study because of the property stated in the following lemma.

\begin{lemma}
\label{lem:UpperBound}
If it is possible to cover a polygon $P$ with $k$ fully visible regions, then $k$ is an upper bound on the number of \boxxes that can be placed in $P$.
\end{lemma}

\begin{proof}
    This follows directly from the definition of fully visible regions, which implies that at most one \boxx can be placed in each fully visible region.
\end{proof}

In the next lemma, we give a simple criterion to determine when a convex polygon is fully visible.

\begin{lemma}
\label{lem:FullyVisible}
    Let $C\subseteq P$ be a closed convex polygon.
    The set $C$ is fully visible if and only if no edge of $C$ contains any polygon vertex of $P$ in its interior.
\end{lemma}

\begin{proof}
    We show the first direction by contraposition.
    If there is a vertex $v$ of $P$ on an edge of $C$ then we can take two points $p,q$ on the same edge left and right of $v$. 
    The two points cannot see each other and thus $C$ is not fully visible.
    
    For the reverse implication, let $p,q\in C$ be 
    two \boxxes. 
    By convexity, the line segment $p,q$ is completely contained in $C$ as well as $P$. 
    Additionally, by our restriction of the allowed edges, the open segment $p,q$ cannot contain a vertex of $P$. 
    Thus, by the definition of open-visibility, $p$ and $q$ are visible to each other.

    Note that a vertex can be between two adjacent collinear edges if one edge is open and one edge is closed.
\end{proof}

Next, we investigate under which conditions a \boxx{} in a fully visible region $C$ is a spectator.
For this, we need to define the notion of minimal visibility.
We say that a point $p\in C\subseteq P$ is  \textit{vision minimal}, if for every other point $q \in C$, it holds that the visibility region of $p$ is contained in the visibility region of $q$,
i.e., $\forall q \in C : \Vis(q) \subseteq \Vis(p)$. 
Such a vision minimal point is also called a \textit{spectator}.
Note that a vision minimal point is not always unique, for example, if the polygon $P$ is convex.

The following observation tells us that we can always place a \boxx with a spectator. 
We say that a point $p\in P$ is \textit{essentially fixed}, if for any solution $B$, we can remove one of the \boxxes $q$ and replace it by the point $p$.
The next observation follows from the definition.

\begin{observation}
    \label{obs:vision-minimal-replacement}
    Let $p$ be a spectator with respect to the fully visible region $C = C_1$.
    Furthermore assume that $C_1,\ldots,C_k$ are fully visible regions that cover $P$ and $k$ is the threshold for the \nestingbirdbox problem.
    Then $p$ is essentially fixed.
\end{observation}

\begin{proof}
    Let $B$ be any solution. As the threshold is $k$ we need at least $k$ \boxxes. Furthermore, due to the convex cover we must have at least one \boxx in each region $C_i$.
    Thus there must be also a \boxx  $q$ in $C_1$.
    By definition, the visibility region of $q$ is contained in the visibility region of $p$, i.e., $\Vis(q)\subseteq \Vis(p)$. 
    Thus $B$ with $q$ replaced by $p$ will also be a valid solution.
    And this finishes the proof.
\end{proof}

We may have two regions that are not vision minimal, but we still want to enforce that \boxxes are placed in a specific position.
For example, in \Cref{fig:spectator-idea}, we see that the visible region from the two nooks can be covered by two fully visible regions.
As such, we know there may be at most two guards in the shared visible region of these two points.
Furthermore, this also implies that whenever there are two guards in this region, they can always be moved towards the corners in the nooks.

Let $\Vis(P)$ for a set of points $P$ be the union of the visibility regions of all points in $P$.
To formalize this, we say that a pair $P = \{ p, p' \}$ of points is vision co-minimal when
\[
    \forall q \in \Vis(p, p'), \forall q' \in \Vis(p, p') : q \text{ does not see } q' \implies \Vis(p, p') \subseteq \Vis(q, q').
\]
\begin{lemma}
    A pair of mutually non-visible points $p, p'$ is vision co-minimal if two fully visible regions cover $\Vis(p, p')$.
\end{lemma}
\begin{proof}
    Note that $p \not\in \Vis(p')$, so $\Vis(p, p')$ can never be covered by a single fully visible region.
    Let $C_1$ and $C_2$ be those two fully visible regions.
    We know by \Cref{lem:UpperBound} that at most two \boxxes can be placed in $\Vis(p, p')$, one in $C_1$ and one in $C_2$.
    Suppose a \boxx $q$ is placed in $C_1$ and a \boxx $q'$ in $C_2$.
    Then, $C_1 \subseteq \Vis(q)$ and $C_2 \subseteq \Vis(q')$ by the definition of fully visible regions.
    As such, we get that $\Vis(p, p') = C_1 \cup C_2 \subseteq \Vis(q, q')$.
    This shows that $p, p'$ are vision co-minimal if $\Vis(p, p')$ can be covered by two fully visible regions.
\end{proof}

As we can see on the right of Figure~\ref{fig:spectator-idea}, we might have a region that is not vision \mbox{(co-)minimal}, but we still want to enforce a \boxx to be placed there.
To this end, we define \textit{relative vision minimality}.
Intuitively, relative vision minimality ensures that a point is vision minimal when you remove the visibility regions from points already known to be vision minimal.
Formally, we define the notion inductively.
Let $P_1,\ldots, P_k$, where $P_i = \{ p_i \}$ or $P_i = \{ p_i, p_i'\}$ and $C_1, \ldots, C_k$ be their visibility polygons.
For $i=1$, relative vision minimality is the same as normal vision minimality.
For $i > 1$, we say that $P_i = \{ p_i \}$ is \textit{relative vision minimal} if the following condition is met.

\[\forall q\in C_i\setminus \Vis(p_1,\ldots,p_{i-1}): \Vis(p_i)\setminus \Vis(p_1,\ldots,p_{i-1}) \subseteq  \Vis(q) \setminus \Vis(p_1,\ldots,p_{i-1}).\]
Likewise, if $P_i = \{p_i, p'_i\}$, it is \textit{relative vision minimal} if

\begin{flalign*}
&\forall q\in C_i\setminus \Vis(p_1,\ldots,p_{i-1}), \forall q' \in C_i'\setminus \Vis(p_1,\ldots,p_{i-1}) : \\
&\qquad q \text{ does not see } q' \implies \Vis(p_i, p'_i)\setminus \Vis(p_1,\ldots,p_{i-1}) \subseteq  \Vis(q, q') \setminus \Vis(p_1,\ldots,p_{i-1}).&
\end{flalign*}

In other words, the visibility region of $p_i$ is minimal, once we ignore the visibility regions of the previous \boxxes.





\begin{lemma}
    \label{lem:SpectatorReplacement}
    Given sets of vision co-minimal spectators $P_1, P_2, \ldots,P_k$ and their respective visibility polygons $C_1,\ldots,C_k$ in the region $Q$, the positions of $P_1,\ldots,P_k$ are essentially fixed.
\end{lemma}

\begin{proof}
    Let $B$ be any solution. 
    There are exactly $|P_i|$ \boxxes of $B$ inside $C_i$, which we denote as $Q_i$.
    Note that 
    \[\Vis(P_1,\ldots,P_k) \subseteq \Vis(Q_1,\ldots,Q_k).\]
    Thus 
    \[B' = B \setminus \{Q_1,\ldots,Q_k\} \cup \{P_1,\ldots,P_k\}\]
    is also a solution.
    This shows the claim.
\end{proof}

For the remainder of the proof, we first prove that we are able to place a blocker in the central chamber in \Cref{sub:blockers}, such that there is a \boxx within the blocker covering a region to the right of the blocker until the right edge of the central chamber.
Then, we prove the central chamber can be fully covered by spectators in \Cref{sub:specs}, before finally describing the gadgets for which we placed blockers in the first place, in \Cref{sub:gadgets}.
We chose this order as the ability to cover the central chamber with spectators relies on the construction of the blockers.
Likewise, the ability to have spectators guard the auxiliary regions of the gadgets depends on whether the central chamber is covered by spectators.
In each step, we pay attention to where we may assume spectators are located and which regions then must be covered by the spectators.

\subsection{Blockers}
\label{sub:blockers}

\begin{figure}
    \centering
    \includegraphics[page=1, width=0.5\linewidth]{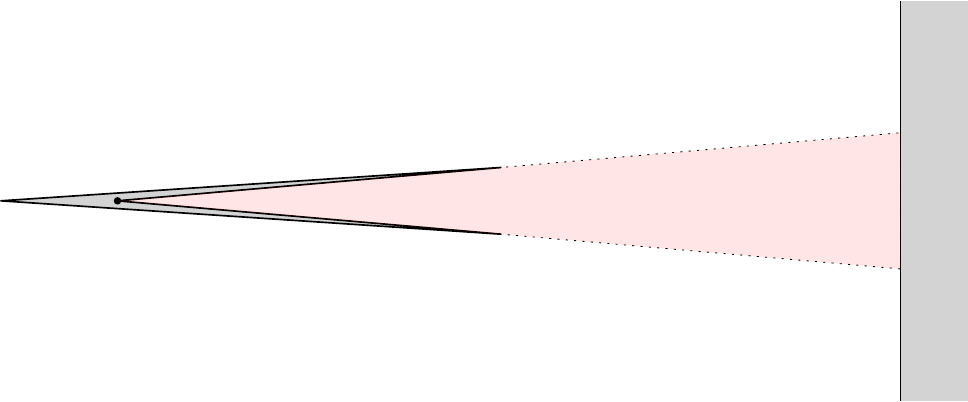}
    \caption{A blocker. Note that it can be made arbitrarily thin.}
    \label{fig:blocker}
\end{figure}

In this section, we present a gadget called a blocker.
A blocker is an arbitrarily thin hole in the polygon with a nook on its right-hand side.
All blockers will appear in the central chamber of the polygon.
An example of a blocker is depicted in \Cref{fig:blocker}.

First, we define what we mean by the central chamber and a blocker formally.

\begin{definition}[Central Chamber]
    The central chamber is a rectangular axis-aligned region with the following edges:
    \begin{itemize}
        \itemsep0em 
        \item \textbf{$e_\texttt{top}$}: an edge segmented by constraint gadgets.
        \item \textbf{$e_\texttt{left}$}: an edge segmented by spectator gadgets.
        \item \textbf{$e_\texttt{bottom}$}: an edge segmented by variable gadgets.
        \item \textbf{$e_\texttt{right}$}: a continuous edge without any gadgets.
    \end{itemize}
\end{definition}

\begin{definition}[Blocker]
    A blocker $B$ is a hole inside the central chamber consisting of three reflex vertices and a convex vertex $b$.
    The visibility region $\Vis(b)$ does not contain any polygon edges outside of its own edges and the edge $e_\texttt{right}$ of the central chamber.
    The visibility region $\Vis(b)$ must include the rightward axis-aligned ray originating from $b$ until the point where it intersects $e_\texttt{right}$.
    We ensure that the y-coordinate of $b$ is the same as the y-coordinate of the left reflex vertex and call this the y-coordinate of $B$, $y(B)=y(b)$.
\end{definition}

We point out that any blocker $B$ can be set to make $\Vis(b)$ arbitrarily narrow.
We claim that there must always be a spectator within the nook of a blocker.
\begin{lemma}
    Let $B_1, B_2, \dots, B_k$ be all $k$ blockers in the polygon.
    Let $b_i$ for $1 \leq i \leq k$ be the point at the convex vertex of $B_i$ and let $C_i = \Vis(p_i)$.
    Then there are essentially-fixed spectators at $b_1, b_2, \dots, b_k$.
\end{lemma}
\begin{proof}
    For any blocker $B_i$, $b_i$ can see the edges of the blocker incident to it.
    The rays of $b_i$ along the incident edges must intersect only $e_\texttt{right}$, by the definition of blockers.
    As $e_\texttt{right}$ is a singular edge, this means that the visibility polygon $\Vis(b_i)$ is a triangular region consisting of a line segment along $e_\texttt{right}$, and both rays along the incident edges.
    Clearly, any two points within this region can see each other, so $\Vis(b_i)$ is a fully visible region.
    Then, by \Cref{lem:SpectatorReplacement}, there is an essentially-fixed spectator at $b_i$.
\end{proof}

\subsection{Central Chamber Spectators}
\label{sub:specs}

\begin{figure}
    \centering
    \includegraphics[page=7]{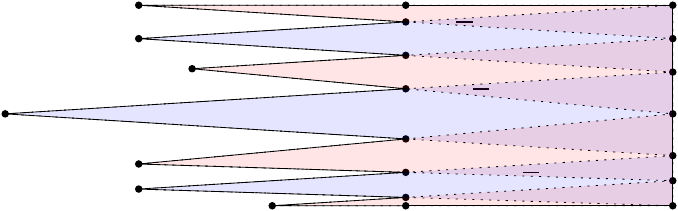}
    \caption{The construction of the triangular pockets that contain spectators at their convex vertex.
    Here, the bold line segments in the central chamber are arbitrarily thin blockers.}
    \label{fig:centralchamber2}
    \label{fig:centralchamber3}
\end{figure}
The central chamber is a large rectangular region with a set of blockers in its interior.
Per \Cref{sub:blockers}, each of the blockers must have an essentially-fixed spectator on its convex vertex.

We create a set of triangular pockets along the left edge of the central chamber.
In each pocket, our aim is to have a single spectator at its convex vertex.
Together, the visibility polygons of these spectators, along with those inside the blockers, should cover the entire central chamber.
Furthermore, only points on the right edge of the central chamber may be seen by multiple left-edge spectators.
All internal points must be seen by exactly one or two spectators.

To create these pockets, we do the following.
First, we select a set of points on the left edge that will be the reflex vertices between two pockets.
Notably, for each blocker, we select the point at the same y-coordinate as the blocker.
Additionally, pick some point in between each reflex vertex corresponding to a blocker such that no two consecutive reflex vertices correspond to a blocker.
Then, for every pair of consecutive reflex vertices on the left edge, we also pick a point on the right edge in between their y-coordinates. 
An example of picking such points is shown in \Cref{fig:centralchamber2}.

Now, we create a set of triangular pockets using the following process:
For each point on the right edge, draw a line through it and the two points on the left edge that have the closest y-coordinate above and below it.
Furthermore, draw a line through the top and bottom edge of the central chamber.
This creates a pattern of triangular pockets along the left edge, such that the visibility polygon has vertices at the chosen points on the right edge of the central chamber, as shown in \Cref{fig:centralchamber3}.

\begin{lemma}
\label{lem:centralchamber}
    Given a polygon $P$ with a central chamber and a set of blockers $B$.
    Let points $p_1, p_2, \dots, p_{k+1}$ be points on $e_\text{left}$ with $y(e_\text{top})=y(p_1)< y(p_2) < \dots < y(p_k+1)=y(e_\text{bottom})$ and points $y(e_\text{top})=q_0, q_1, q_2, \dots, q_{k+1}, q_{k+2}=y(e_\text{bottom})$ be on $e_\text{right}$ with $y(q_0)=y(q_1) < y(q_2) < \dots < y(q_k)=y(p_{k+1})$, where for $1 \leq i \leq |B|$, $y(p_{2\cdot i})=y(B)$.
    Let $l(p_i, q_j)$ be the unique line through $p_i$ and $q_j$.
    
    For each $1 \leq i < k$, the polygon contains a nook alongside $e_\text{left}$ consisting of the line segments $l(p_i, q_{i-1})$, $l(p_{i+1}, q_{i+1})$, ranging from $e_\text{left}$ to their intersection.
    Then, there is an essentially-fixed spectator at the convex vertex at the intersection of $l(p_i, q_{i-1})$ and $l(p_{i+1}, q_{i+1})$.
\end{lemma}
\begin{proof}
    
Let $r$ be the convex vertex in one of the triangular pockets along the left edge.
It is clear that a spectator would be essentially fixed at $p$ if there are no blockers incident to its visibility polygon.
In that case, its visibility polygon is simply a triangular region in its pocket and the central chamber.

However, by construction, each blocker is at the same y-coordinate as a reflex vertex between pockets.
As such, spectators in two consecutive pockets can each see part of the blocker.
Relative to the spectator in the blocker being essentially fixed, these two spectators are vision co-minimal, as the union of their visibility polygons can be covered by two fully visible regions.
Specifically, by the region with y-coordinate greater than the reflex vertex between them, and the region with y-coordinate smaller than said vertex.
Note that no other blockers may be adjacent to their visibility polygons by construction.
As such, all the spectators in the pockets along the left edge are essentially fixed.
\end{proof}

\begin{figure}
    \centering
    \includegraphics[page=22]{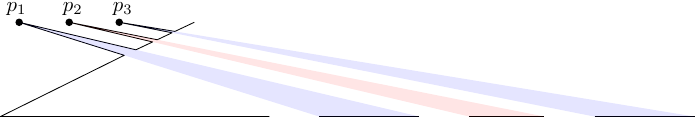}
    \caption{Ensuring all of $e_\texttt{top}$ and $e_\texttt{bottom}$ are covered. The spectator at $p_1$ covers the first edge, the spectator at $p_2$ the second edge, etc. This construction can be made arbitrarily thin.}
    \label{fig:guard_edges}
\end{figure}

Now, to ensure that the entire central chamber is covered, we only need to ensure that we guard the top and bottom edge of the central chamber.
While the outermost spectators of the triangular gadgets on the left can see these edges in a vacuum, they cannot if we add intermittent gadgets to these edges.
Adding gadgets creates new vertices on these edges that block visibility.
As such, for each line segment, we add a small new triangular pocket within the top or bottom pocket on the left edge respectively, such that a spectator in these pockets will always guard the line segment.
An example of one these pockets is displayed in \Cref{fig:guard_edges}.
Note that we ensure the vertices of the edges are always guarded from within the gadget; for our purposes it is enough to guard the edge.
Finally, we ensure the edge with the small pockets is slightly curved.
This way, the vertex in its corner can see the entire edge.
It does make its visibility region not convex.
However, as the points in the small pockets ($p_1, p_2, p_3$, etc.) are essentially fixed, it will be essentially fixed as well.

\section{Gadgets}
\label{sub:gadgets}

For the gadgets, we may now assume the central chamber is fully covered.
We now ensure that each gadget is entirely covered (barring any \boxx segments) and that if some \boxx sees outside of the gadget, it only sees to the right, where it can never observe any other \boxxes.

\subsection{Variable Gadget}
\label{sub:VariableGadget}

\begin{figure}[t]
  \centering
    \centering\includegraphics[page=10]{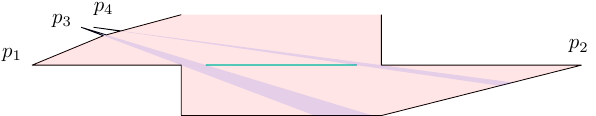}
    
  \caption{The basic construction of the variable gadget with four spectators and one \boxx segment.
  For each spectator, its visibility region is shown.
  Note that the \boxx segment can be put arbitrarily close to the top edge of the gadget, in case it needs to be seen from another gadget at an angle.}
  \label{fig:varGadget} 
\end{figure}

In this section, we are describing the variable gadget. We will construct a small part of a polygonal domain such that there is a \textit{\boxx segment} (or variable segment). 
The \boxx segment contains without loss of generality at least one \boxx in an optimal \boxx placement. 
Furthermore, the \boxx is free to be anywhere on the \boxx segment. We can then use these \boxx segments to encode variables that are free to take any value within a specific range. 

\textbf{Description.}
    The idea of the construction is to enforce that in every optimal \boxx placement, there is one \boxx placed on the \boxx segment, which then represents one variable taking a specific value within its range. 
    This is achieved by having spectators such that the union of their visibility regions includes the whole gadget except for the \boxx segment. 
    This way, the final \boxx in the gadget must be placed on the \boxx segment. The construction is visualized in Figure~\ref{fig:varGadget}. 
    The spectators are on the convex vertices in the four triangular pockets. 
    The boundary of the polygonal domain is black and the union of the four visibility polygons are shown in orange, while the \boxx segment is shown in turquoise.
    We point out that this construction is highly flexible, as the \boxx segment can be put at a different rotation or be made longer or shorter, depending on the positions of the triangular pockets.
    The threshold for this gadget equals $5$.

    We interpret the value of the left most point of the \boxx segment as $-1$, the point of the right most point as $+1$ and any point is defined through linear interpolation, see Figure~\ref{fig:varGadget}.

  

\begin{lemma}
    The variable gadget has a threshold $t(G_\texttt{variable}) = 5$, i.e., if there are five spectators in the gadget from \Cref{fig:varGadget}, then there must be a spectator on the variable segment.
\end{lemma}
\begin{proof}
    Figure~\ref{fig:varGadget} describes the polygonal region $Q$, the points $p_1,\ldots,p_4$ and their visibility polygons $C_1,\ldots,C_4$.
    We see that $p_2$ is vision minimal.
    Furthermore, given $p_2$ is already essentially fixed and that the central chamber is fully covered, then $p_3$ and $p_4$ are vision minimal and become essentially fixed as well.
    Finally, due to the slight curvature of the top left edge, $p_1$ only becomes essentially fixed after all other \boxxes are essentially fixed.
    This curvature ensures that $p_1$ covers all the edges and vertices not covered by another \boxx.
    Afterward, the only region of the gadget that is left uncovered is the \boxx segment.
    As such, a fifth \boxx must be placed upon the \boxx segment and can thus be seen as encoding a number in $[-1, 1]$.
    As any \boxx on the \boxx segment will contain the entire \boxx segment in its visibility polygon, the gadget may contain at most five \boxxes.
    So, we see the threshold $t(G_\texttt{variable}) = 5$.
\end{proof}

We will not only use the variable gadget along $e_\text{bottom}$, but also within other gadgets to encode meta variables used to encode specific operations.
In this context, we will call the variable gadget a \textit{critical segment} instead.

\subsection{Copy and Scaling Gadget}
\label{sub:Scaling}

\begin{figure}[t]
    \centering
    \includegraphics[page =13]{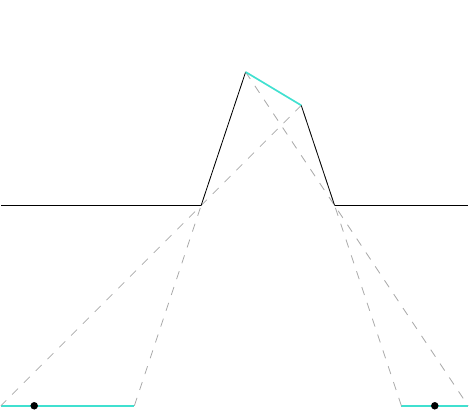}
    \caption{The $\leq$-scaling gadget is a nook in the upper part of the polygon. 
    It consists of two pivot points and a critical segment. 
    The gadget is connected to two \boxx segments.}
    \label{fig:copy-labels}
\end{figure}

Next, we will describe the copy and scaling gadget. 
Note that the same gadget can be used for copying a variable and scaling a variable.
We need this gadget in order to be able to have multiple instances of the same variable, e.g. if it appears in multiple constraints and has to obey all of them at the same time. 
Additionally, this gadget allows for scaling of the variable output. 
The basic idea is having two \boxx segments at the bottom of the polygon, one for the original variable $x$ and one for the copied variable $x'$. We then need two separate gadgets at the top of the polygon to make sure that $x\leq x'$ and $x\geq x'$, respectively. Given the value of $x$, we use these gadgets to guarantee that a box can be placed in both of them iff $x'=x$.

We will also show how to construct a gadget that can scale a variable, i.e., it enforces $\alpha x \leq y$ or $\alpha x \geq y$.
We begin with a sketch of the gadget, which is useful as it highlights specific aspects and notation and neglects details that we will fill in later.
We want to point out that the versions for ``$\leq$'' and ``$\geq$'' look different.

\textbf{Components and Description of the Copy and Scaling Gadget.}
As mentioned already, this gadget is split into the $\leq$-scaling gadget and the $\geq$-scaling gadget.
To get a $=$-scaling gadget, we simply connect a pair of \boxx segments to both a $\leq$-scaling gadget and a $\geq$-scaling gadget.
We assume the left \boxx segment represents the variable $x$ and the right \boxx segment represents the variable $y$.

We first present the common setup for both the $\leq$-scaling gadget and the $\geq$-scaling gadget.
Then, we will discuss how the two deviate and the impact on the gadget.

Both gadgets have two pivot points on $e_\text{top}$ with x-coordinates such that they are always in between the \boxxes representing $x$ and $y$, regardless of which values $x$ and $y$ take.
Note that, because the two pivot points are on $e_\text{top}$, they are on a line parallel to the \boxx segments of $x$ and $y$.
We also create a critical segment whose endpoints are constructed as follows.
Take the endpoint of the \boxx segments and shoot rays through the  pivot points.
The intersection points of the rays define the critical segment.

The two gadgets differ in the direction from which vision is obscured with regards to the pivot points.
If the region between the pivot points is within the polygon, we have a $\leq$-scaling gadget, as depicted in \Cref{fig:copy-labels}.
Likewise, if the region between the pivot points is a hole in the polygon, we have a $\geq$-scaling gadget, as depicted in \Cref{fig:copy-labels-reverse}.
We note that the illustrations in the figure also display that everything but the critical segment will be covered by the visibility polygons of the spectators.

\begin{figure}[tph]
    \centering
    \includegraphics[page=23]{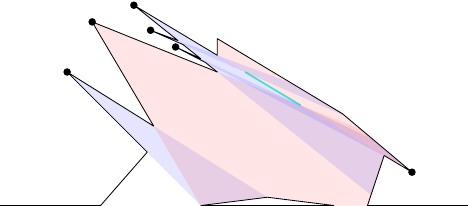}
    \caption{Construction of the $\geq$-gadget. The black disks represent the position of the spectators, which we denote by $p_1, p_2, \dots, p_6$ in counterclockwise order.}
    \label{fig:copy-labels-reverse}
\end{figure}


\textbf{Calculation.}
We will now prove that if the rays 
from $x$ and $y$ through $p$ and $q$ intersect on the critical segment, then $x$ and $y$ encode the same value.
To this end, we say $x_{i}$ is the point at which $x$ would encode $-1 \leq i \leq 1$.
We define $y_i$ analogously.
\Cref{fig:copy-calculations} displays the gadget schematically, where $p$ and $q$ are the pivot points and $s$ and $t$ are the endpoints of the critical segment.

\begin{figure}[th]
    \centering
    \includegraphics[page = 16]{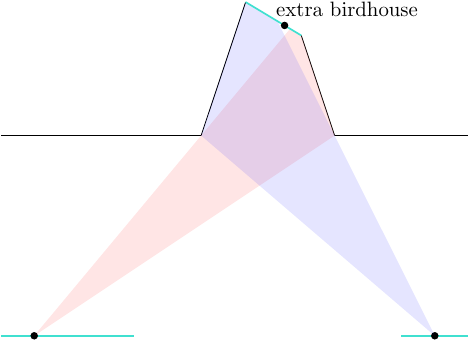}
    \caption{If the \boxxes are at a good position, we can place an additional \boxx{} inside the gadget.}
    \label{fig:visibility-1}
\end{figure}

\begin{lemma}[{\cite[Lemma 27]{AAM22}}]
\label{lem:calc_copy}
    It holds that $\frac{{\|x_{-1}-x\|}}{\|x_{-1}-x_{1}\|} = \frac{{\|y_{-1}-y\|}}{\|y_{-1}-y_{1}\|}$.
\end{lemma}
\begin{proof}
    Note that the statement holds if $\frac{x_{-1}x}{xx_1} = \frac{y_{-1}y}{yy_1}.$
    In this proof, we denote the distance between two points $u$, $v$ by $uv$ instead of $\| u-v \|_2$ for improved readability.
    
    We apply the intercept theorem to the pivot points $p$, $q$, $s$, and $t$ to get the following equations:
    \[\frac{x_{-1}x}{xx_1} = \frac{ce}{ac}, \quad
    \frac{y_{-1}y}{y_{-1}y_{1}} = \frac{cd}{bc}, \quad
    \frac{de}{cd} = \frac{pq}{pr} = \frac{ac}{bc} \]

Using the four equations, we get
\[\frac{x_{-1}x}{xx_1} 
= \frac{ce}{ac}
= \frac{ce \cdot cd}{ac \cdot cd}
= \frac{pr \cdot cd}{ac \cdot qr}
= \frac{ac \cdot cd}{ac \cdot bc}
= \frac{cd}{bc} 
= \frac{y_{-1}y}{yy_1}.\]

This finishes the proof.
\end{proof}

\begin{figure}[t]
    \centering
    \includegraphics[page =3]{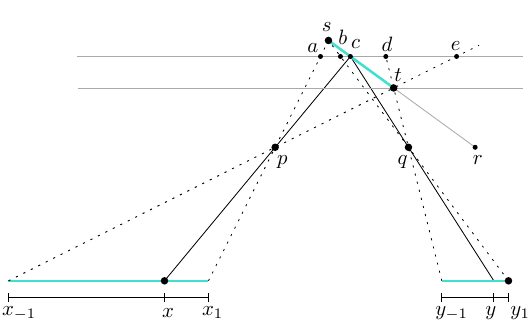}
    \caption{The figure contains the points and length to show that the scaling gadget enforces a linear constraint.}
    \label{fig:copy-calculations}
\end{figure}

Now assume that we know a \boxx must be placed on the critical segment and consider the $\leq$-scaling gadget. By Lemma~\ref{lem:calc_copy} and the fact that we can always move the \boxxes further apart from each other to create more space on the critical segment, it must indeed hold that $x\leq x'$. Looking at the $\geq$-scaling gadget on the other hand, by the same logic and the fact that we can move the two \boxxes closer together to create more space on the critical segment, $x\geq x'$ must be true in this case.

\begin{lemma}
    The scaling gadget has a threshold $t(G_\texttt{scaling}) = 8$, i.e., if the variable segments for variables $x$ and $y$ are both connected to the $\leq$-gadget and the $\geq$-gadget, as in \Cref{fig:copy-labels} and \Cref{fig:copy-labels-reverse} respectively and there are eight \boxxes in these gadgets, then $x = y$.
\end{lemma}
\begin{proof}
    By \Cref{lem:calc_copy}, we see that if a \boxx is on the critical segment of both the $\leq$-gadget and the $\geq$-gadget, then $x \leq y \leq x$, so $x = y$.

    As such, we simply need to show that there must be a \boxx on each critical segment.
    Let $p_1, p_2, \dots, p_6$ denote the vertices denoted by black disks in \Cref{fig:copy-labels-reverse} in the $\geq$-gadget in counterclockwise order (starting from the left).
    We observe that $p_1$, $p_3$, $p_4$, and $p_6$ are vision minimal and thus contain essentially-fixed spectators.
    Furthermore, $p_2$ and $p_5$ are essentially fixed when the central chamber is completely covered by \Cref{lem:centralchamber} and $p_1$, $p_3$, $p_4$ and $p_6$ are essentially fixed.
    On the other hand, the $\leq$-gadget does not contain any essentially-fixed spectators.

    As such, there are two \boxxes remaining.
    Notice that, assuming the variable segments of $x$ and $y$ both contain a \boxx, the entire $\geq$-gadget is covered by visibility polygons except the critical segment.
    Furthermore, the $\leq$-gadget is convex and can never contain more than a single \boxx.

    As such, there must lie a vertex on the critical segment of the $\geq$-gadget, which enforces $x \geq y$.
    By \Cref{lem:calc_copy}, a \boxx can be placed on the $\leq$-gadget's critical segment if and only if $x \leq y$.
    Additionally, if $x \leq y$, no \boxx can be placed in the $\leq$-gadget at all.
    Then, it must be that both $x \leq y$ and $y \leq x$, so $x = y$ if the $\leq$-gadget and the $\geq$-gadget contain eight \boxxes in total.
\end{proof}

\subsection{Addition Gadget}
    \begin{figure}[tbhp]
        \centering
        \includegraphics[page = 17]{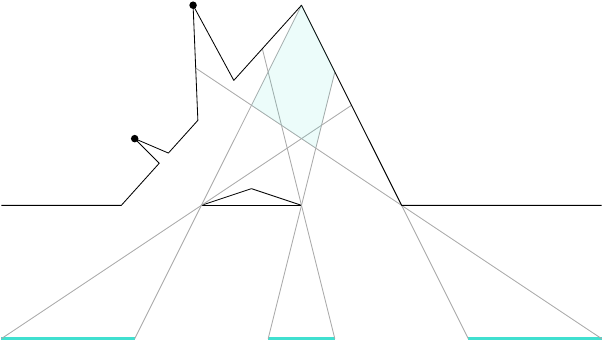}
        \caption{The $\geq$ addition gadget. The reverse gadget is the exact mirror image. The black disks represent the position of the spectators.}
        \label{fig:addition-labels}
        \label{fig:addition-reverse}
    \end{figure}

    First, we describe a gadget enforcing $x+y\geq z$.
    We call this the $\geq$-addition gadget.
    Thereafter, we describe the very similar 
    $\leq$-addition gadget enforcing $x+y\leq z$.
    
\textbf{Description of the $\geq$-Addition Gadget.}
    We want to point out that our addition gadget follows the ideas of Abrahamsen, Adamaszek and Miltzow closely~\cite{AAM22}.
    We repeat many of their ideas for the sake of self-containment and integration.
    The following description is illustrated in Figure~\ref{fig:addition-labels}.
    We assume that there are three variable segments representing the variables $x$, $y$, and $z$.
    This can be done using the scaling gadget as described in Section~\ref{sub:Scaling}.
    We call those three variable segments left, middle and right variable segment.
    We set the middle variable segment to be two times shorter than the others.
    The distance between the left and the middle variable segments is the same as the distance between the middle and the right variable segments.
    Furthermore, the points $p_1$, $p_2$, and $p_3$ are on $e_\text{top}$, are equidistant, and $p_2$ has the same x-coordinate as $z_0$.

    


    Instead of a critical segment, this gadget uses a \textit{critical region}.
    Unlike the critical segment, the critical region is a polygonal region.
    However, it is still critical in the sense that we may only place a \boxx in the region if and only if some constraint holds.
    Specifically, the critical region of the gadget is exactly the region for which variable assignments of $x$, $y$, and $z$ exist such that it is covered by none of $\Vis(x), \Vis(y)$, or $\Vis(z)$, nor any essentially-fixed spectators.

\textbf{Reverse Gadget.}
    Now, we shortly describe the reverse addition gadget enforcing the constraint $x+y \leq z$. 
    The gadget is the mirror image of the figure illustrated in Figure~\ref{fig:addition-reverse}.
    The gadget is a vertical mirroring of the normal gadget.
    In this section, we will usually only consider the $\leq$ addition gadget, as all arguments apply symmetrically to the $\geq$ addition gadget.
    
\textbf{Visibility of the $\geq$-Addition Gadget.}
    Let $p$, $q$, and $r$ be \boxxes on the left, middle and right variable segments.
    The more $p$ moves to the right (the larger the value it represents), the smaller the critical area that it sees.
    It sees the entire critical area at the left most position and nothing of the critical area at the right most position.
    It is reversed for the points $q$ and $r$. 
    The further $q$ or $r$ is moved to the right (the larger the value it represents), the \textit{more} they see of the critical area.
    Again at the endpoints, they see either nothing or everything of the critical area.
    Furthermore, if there is a spectator at the intended position and \boxx on the three bird segments then the entire gadget is covered, 
    except potentially some part of the critical area.

\textbf{Calculation.}
In order to study the constraint enforced by the gadget, we consider the situation where exactly one point, denoted by $t$, is not seen by the three \boxxes in the critical area.
We define the following values.
\begin{description}
    \item[$p, q, r:$] The \boxxes on the variable segments representing the values $x$, $y$, and $z$ respectively.
    \item[$C$:] We denote the distance between the midpoints between the left and middle bird segment by $C$.
    Note that this is the same as the distance between the middle and right bird segment by construction.
    \item[$s,S$:] the distance between $t$ to $p_2$ and $p_2$ to $r$ respectively.
    \item[$D_1$, $D_2$:] The distance from $p$ to $r$ is $D_1$ and the distance from $q$ to $r$ is $D_2$.
    \item[$\delta$:] We denote the distance between the left and the middle pivot points by $\delta$.
    Note that this equals the distance between the middle and the right middle point by constructions.
\end{description}

\begin{lemma}
\label{lem:calc_add}
    If exactly one point in the critical region is not covered by the visibility polygons of $x$, $y$, and $z$, then $x+y=z$.
\end{lemma}
\begin{proof}
    If exactly one point $t$ in the critical region is not covered by the three visibility polygons, the rays from $p$, $r$, and $q$ through $p_1$, $p_2$, and $p_3$ respectively must intersect in $t$.
    By the intercept theorem, we see that $\frac{s}{S} = \frac{\delta}{D_1}$ and $\frac{s}{S} = \frac{\delta}{D_2}$.
    Together, they imply $D_1 = D_2$.
    Then, $D_1 = C-x +z/2$ and $D_2 = C -z/2 + y$, so \[C-x +z/2 = D_1 = D_2 = C -z/2 + y \Rightarrow -x + z/2 = -z/2 + y \Rightarrow x+y = z. \qedhere\]
\end{proof}

Note that this is for the case where the three visibility lines meet in a point in the critical area.
If the three do not meet in a single point, we can place a \boxx at any point $t$ in the critical region that is not covered.
Let this point be $t'$ and let $x'$, $y'$, and $z'$ be the values of the variables of the \boxxes on the three variable segments, such that their visibility lines meet in point $t'$.
Then, for the $\leq$ addition gadget, we see that $x' \leq x$, $y' \leq y$, and $z' \geq z$.
So, we see that $x + y \leq z$.
Likewise, for the $\geq$ addition gadget, it follows that $x' \geq x$, $y' \geq y$, and $z' \leq z$, and thus $x + y \geq z$.

As $x,y$ could be larger or $z$ could be made smaller and we would still see at least some point in the critical area, it holds that this gadget enforces 
$x +y \geq z$.

\begin{figure}
    \centering
    \includegraphics[page = 19]{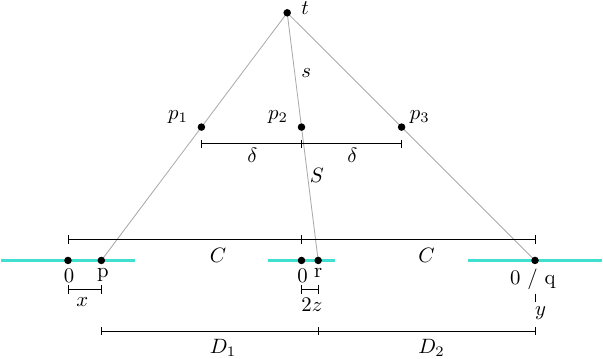}
    \caption{Illustration of different lengths in the addition gadget.}
    \label{fig:addition-calculation}
\end{figure}

\begin{lemma}
    The addition gadget has a threshold $t(G_\texttt{addition}) = 6$, i.e., if the variable segments for variables $x$, $y$, and $z$ are connected to a $\leq$ or $\geq$ addition gadget, as in \Cref{fig:addition-calculation}, and the gadgets contain six \boxxes total, then $x + y \leq z$ (or $x + y \geq z$).
\end{lemma}
\begin{proof}
    We note that the addition gadget (minus a small part that must always be covered by a variable gadget) can be covered by three convex regions, so no more than three \boxxes can occupy the gadget.
    Note that this holds for both the $\geq$ and $\leq$ variants, as they are symmetric.
    Furthermore, there will be an essentially-fixed spectator at the vertices inside the two triangular pockets, as they are vision minimal.
    So, we may assume that there are \boxxes at $p_1$, $p_2$ and each of the three variable segments.
    By \Cref{lem:calc_add}, we can only fit a third \boxx in the gadget if $x + y \leq z$ (or $x + y \geq z$ respectively).
\end{proof}

\subsection{Curved Gadget}
\begin{figure}
    \centering
    \includegraphics[page = 20]{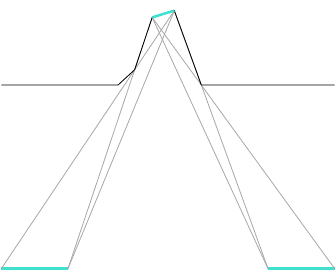}
    \hfill
    \includegraphics[page = 21]{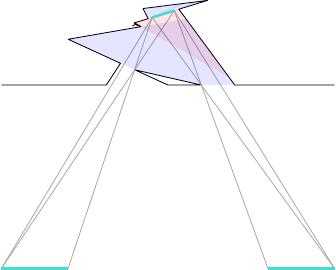}
    \caption{The curved gadget with the visibility regions of the spectators in the gadget. Note that the two pivot points are at different heights. Left: $\leq$-gadget with variable segments shown. Right: $\geq$-gadget with visibility regions of the spectators shown.}
    \label{fig:curved-gadget}
\end{figure}
As described in \Cref{subsec:ccsp}, the curved gadget should ensure that for two \boxxes on equally-sized variable segments representing $x, y$, $x = \frac{a y + b}{c y + d}$.
To this end, we show that our gadgets induce M\"obius transformations, which are of the desired form.

\textbf{Description of the $\geq$-Curved Gadget.}
Curved gadgets largely resemble scaling gadgets.
The main difference between the two gadgets is that one of the reflex pivot points is above $e_\texttt{top}$, instead of both being on this edge.
Otherwise, the strategy to guard this gadget is largely the same as with the scaling gadget.
In addition, this gadget also has two variants, one for $\geq$ and one for $\leq$.
We emphasize that in the construction, we ensure one variable segment is strictly to the left of the curved gadget, while the other is strictly to the right.

An example of the curved gadget is shown in \Cref{fig:curved-gadget}, while a schematic drawing is shown in \Cref{fig:curved_calc}.
The curved gadget is located on the upper edge of the central chamber, while the variable \boxx segments are on the bottom of the central chamber.
There are two \boxx segments that interact with the curved gadget.
We can send a ray from \boxx $x$ through $p$. 
Let $r$ be the intersection of this ray with the \boxx segment of the curved gadget.
The \boxx may only be at this point on the segment or to the left of it, otherwise it will be in the visible region of $r$.

\textbf{Reverse Gadget.}
Like the scaling gadget, the $\geq$ curved gadget is constructed by adding a hole between the pivot points and ensure that the \boxxes can see past the pivot points on the outer side.
The reverse gadget is displayed in \Cref{fig:curved-gadget} as well, together with the visibility regions of its spectators.

\textbf{Calculation.}
To show the curved gadget indeed encodes a curved constraint, we assume that both variable segments are on the x-axis and thus have a $y$-coordinate of $0$.
Likewise, we assume the critical segment within the gadget lies on the line $\ell : y=kx+l$.
Now, the curved gadget encodes the following construction:
\begin{enumerate}
    \item Draw the line through $x$ and $p$ and let $r$ be its intersection with $\ell$.
    \item Draw the line through $r$ and $q$ and let $y$ be its intersection with the $x$-axis.
    \item Now the curved gadget encodes the constraint $f(x)$, where $f(x)$ is the $x$-coordinate of $y$.
\end{enumerate}

Each step of this process defines a central projection between lines, which are known to be M\"obius transformations, which are of the form $\frac{ax+b}{cx+d}$.
The composition of two M\"obius transformations are also M\"obius transformations, so $f(x)$ must be of the form $f(x)=\frac{ax+b}{cx+d}$.

\begin{figure}
    \centering
    \includegraphics[page=3]{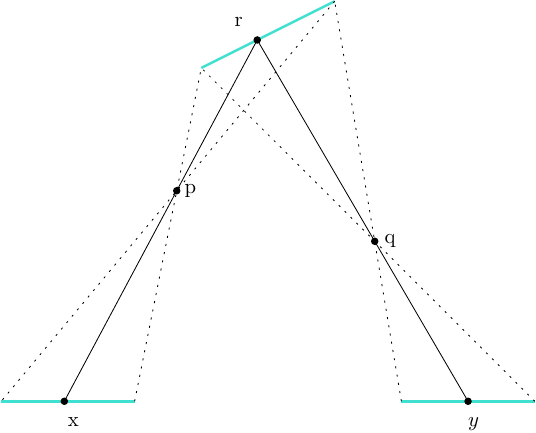}
    \caption{Schematic drawing for the calculation of the curvedness.}
    \label{fig:curved_calc}
\end{figure}
\begin{lemma}
\label{lem:nonaffine}
    $f(x)=\frac{ax+b}{cx+d}$ is not an affine transformation.
\end{lemma}
    \begin{proof}
    We assume $f(x)$ is affine towards a contradiction.
    Since $f$ is affine and maps the \boxx segment of $x$ to the \boxx segment of $y$ (by construction) and both \boxx segments have equal length, it must preserve distances.
    This rules out non-trivial scaling and reflection, so $f$ must be a translation.

    If $f(x)$ is affine, the point $f(x)$ should approach infinity as $x$ approaches infinity.
    (Clearly, $f(x)$ is not constant.)
    However, we see that as $x$ approaches infinity, $r$ approaches $(p_x, kp_x+l)$.
    As the $x$-coordinates of $r$ and $q$ are distinct, we know the line through $r$ and $q$ is not parallel to the $x$-axis, and as such, $f(x)=y$ does not approach infinity.
    This contradicts the assumption that $f(x)$ is affine.
\end{proof}

\begin{lemma}
    The curved gadget has a threshold $t(G_\texttt{curved}) = 7$, i.e., if two variables $x$ and $y$ are linked to both a $\leq$-scaling gadget and a $\geq$-scaling gadget (albeit after copying), as in \Cref{fig:curved-gadget} and there are seven \boxxes in these two gadgets combined, then $f(x)=y$ where $f$ encodes some curved constraint.
\end{lemma}
\begin{proof}
    When a \boxx is on the critical segment of both the $\leq$-gadget and the $\geq$-gadget, $f(x)=y$ encodes a constraint of the form $f(x) = \frac{ax+b}{cx+d}$.
    By \Cref{lem:nonaffine}, we know that $f(x)=y$ must be a curved transformation.

    So, we prove that if there are seven \boxxes in these two gadgets, there must be a \boxx on both of the critical segments.
    As with the scaling gadget, we notice that in the $\geq$-gadget, the visibility region of each of the five convex vertices in the nooks are vision-minimal.
    Thus, these are essentially-fixed spectators.

    Then, there remain only two \boxxes to place.
    In the $\geq$-gadget, only the critical segment is uncovered, whereas in the $\leq$-gadget, only a convex region is uncovered.
    (Note that in the $\geq$-gadget, there is also a region uncovered by the spectators, but this region is always covered by the variable \boxxes.)
    Thus, both must contain a \boxx.

    After placing a \boxx on the critical segment in the $\geq$-gadget, we know $x \geq f(x)$.
    Per \Cref{lem:nonaffine}, the $\leq$-gadget will be completely covered by the \boxxes on the variable segments, unless $x = f(x)$, in which case a single point on the critical segment is uncovered.
    As such, there can only be seven \boxxes in these two gadgets if $x = f(x)$ and as such $t(G_\texttt{curved}) = 7$.
\end{proof}

\section{Synthesis and \ER-hardness}
\label{sub:synthesis}

In this section, we synthesize the ideas present in \Cref{sec:overview} and \Cref{sub:gadgets} to prove \Cref{thm:birds}.
We create an instance of the \nestingbirdbox problem that has a solution if and only if an arbitrary existentially quantified polynomial formula $\phi$ is true.
First, we create a central chamber with height 1 and width 1.
For each variable gadget, we choose a single point on $e_\texttt{bottom}$ to represent it, and likewise for each constraint gadget we choose a single point on $e_\texttt{top}$ to represent it.
Then, we connect each variable gadget to the constraint gadget it is involved in by a line between the respective points.
Note that for the addition and curved gadgets, each involved variable needs to be at a specific distance from the constraint gadget; we achieve this by virtue of copying variables to the desired location.

At this point, we have an arrangement of lines $L$ within the polygon.
We now consider all lines $\bar{L}$ between constraint gadgets and variable gadgets not drawn.
All these lines need to be blocked by a blocker, to ensure they do not interfere with gadgets they are not supposed to.
For each line in $\bar{L}$, we choose a point to place a blocker, such that no point is on a line in $L$ and that no two points share the same $y$-coordinate.
Note that at most a polynomial number of constraints exist, each involving a constant number of variables, so we require at most a polynomial amount of blockers.
Likewise, each point can be chosen using only singly-exponentially small coordinates requiring polynomially many bits.

As each gadget can be arbitrarily small, we now increase the width of the gadgets.
This also increases the width of the lines and thus the required width of the blockers.
However, the gadgets and vertices will only be singly-exponentially small, as the points were on a singly-exponentially small grid.

Finally, we add the spectators of the central chamber as described in \Cref{sub:spectators}, such that no spectator sees two blockers and that each blocker is covered by two spectators.
This way, each spectator will be essentially-fixed.

We note that each gadget contains a constant number of vertices and the number of gadgets needed is polynomially bounded by the size of the formula $\phi$.

Finally, we find the total threshold as follows: 
\begin{itemize}
    \item for each variable gadget, we add $t(G_\texttt{variable}) = 5$ to the global threshold; 
    \item for each copy/scaling gadget, we add $t(G_\texttt{scaling}) = 8$ to the global threshold; 
    \item for each addition gadget, we add $t(G_\texttt{addition}) = 6$ to the global threshold; 
    \item for each curved gadget, we add $t(G_\texttt{curved}) = 7$ to the global threshold; 
    \item for each blocker, we add $1$ to the global threshold for the spectator inside the blocker;
    \item for each triangular pocket on $e_\texttt{left}$, we add $1$ to the global threshold;
    \item for each gadget (variable or constraint), we add $1$ to the global threshold for a spectator to cover part of $e_\texttt{top}$ or $e_\texttt{bottom}$.
\end{itemize}
This gives us a global threshold $t$ that ensures we can place $t$ \boxxes if and only if all constraints are satisfied.
This completes the proof that the \nestingbirdbox problem is \ER-complete.
Together with \Cref{sec:membership}, it shows \Cref{thm:birds} and establishes the \ER-Completeness of the \nestingbirdbox problem.

\section*{Acknowledgments}
T.M., J.O., and M.S.~would like to express their gratitude to be invited to the amazing \href{https://ti.inf.ethz.ch/ew/workshops/gwop24/index.html}{GWOP} workshop in Pura, Switzerland, from 3rd to 7th of June 2024, where this work was initiated.

\newpage
\bibliographystyle{alphaurl}
\bibliography{library, ERlibrary}

\end{document}